\documentclass[aps,pra,twocolumn,amsfonts,amssymb,amsmath,showpacs,
floatfix,nofootinbib,groupedaddress,superscriptaddress,citesort]{revtex4}
\usepackage{mathrsfs}
\usepackage{amsfonts}
\usepackage{amstext}
\usepackage{amsmath}
\usepackage{amssymb}
\usepackage[dvipsnames]{xcolor}
\usepackage[dvips]{graphicx}
\usepackage{color}
\def\qed{\leavevmode\unskip\penalty9999 \hbox{}\nobreak\hfill
     \quad\hbox{\leavevmode  \hbox to.77778em{%
              \hfil\vrule   \vbox to.675em%
               {\hrule width.6em\vfil\hrule}\vrule\hfil}}
     \par\vskip3pt}

\newtheorem{theorem}{Theorem}

\newtheorem{proof}{Proof}

\begin{document}

\title{Quantum Discord for Multi-qubit Systems}
\author{Bo Li}
\email{libobeijing2008@163.com.}
\affiliation{School of Mathematics and Computer science, Shangrao Normal University, Shangrao 334001, China}
\author{Chen-Lu Zhu}
\email{zhuchenlu2020@163.com.}
\affiliation{Department of Mathematics, East China University of Technology, Nanchang 330013, China}
\author{Xiao-Bin Liang}
\email{liangxiaobin2004@126.com.}
\affiliation{School of Mathematics and Computer science, Shangrao Normal University, Shangrao 334001, China}
\author{Biao-Liang Ye}
\email{biaoliangye@gmail.com.}
\affiliation{Quantum Information Research Center, Shangrao Normal University, Shangrao 334001, China}
\author{Shao-Ming Fei}
\email{feishm@cnu.edu.cn.}
\affiliation{School of Mathematical Sciences, Capital Normal University, Beijing 100048, China}
\affiliation{Max-Planck-Institute for Mathematics in the Sciences, 04103, Leipzig, Germany}

\begin{abstract}
We evaluate analytically the quantum discord for a large family of multi-qubit states. It is interesting to note that the quantum discord of three-qubits and five-qubits is the same, as is the quantum discord of two-qubits and six-qubits. We discover that the quantum discord of this family states can be concluded into three categories. The level surfaces of the quantum discord in the three categories is  shown through images.  Furthermore, we investigated the dynamic behavior of quantum discord under decoherence. For the odd partite systems, we prove the frozen phenomenon of quantum discord doesn't exist under the phase flip channel, while it can be found in the even partite systems.
\end{abstract}
\maketitle
\section{\bf Introduction}

Quantum correlations are essential features of quantum mechanics which distinguish the quantum from the classical world and play very important roles in quantum information processing.
The quantum correlated states are shown to be more useful than the classically correlated ones in performing communication and computation tasks. Understanding and quantifying various quantum correlations are the primary goals in quantum information theory. The quantum entanglement and nonlocal correlations can be considered the most fundamental
resources in quantum information processing
\cite{Adesso,Ma,Baumgratz,Streltsov,Streltsov1,Styliaris,Xi,Xu,Xiong,Huang,Oszmaniec}, which are tightly
related to quantum coherence \cite{Theurer,Yao,Yu,Jin,zhuhj,xiya}.

The quantum discord is of the most famous quantum correlations proposed by Ollivier and Zurek \cite{Ollivier} and Henderson and Vedral \cite{Henderson}, which quantifies the quantum correlations in bipartite systems without quantum entanglement. It is defined as the minimum difference between the quantum versions of two classically equivalent expressions of mutual information under projective measurements \cite{Rulli,Luo}. Due to the complexity of the minimization process, the computation of quantum discord is a hard task and seldom analytic results are know only for some restricted families of states \cite{Ali,Vinjanampathy,Xie,luo1,lang,ali}. For a bipartite state $\rho$ in systems $A$ and $B$\cite{Ollivier,Henderson}, the quantum discord $D_{A;B} (\rho)$ is defined by
$D_{A;B} (\rho) = \min \limits_{\Pi^{A}}  \Big[ S_{B| \Pi^{A}} (\rho) - S_{B|A} (\rho) \Big],$
where the conditional entropy $S_{B|A} (\rho)=S(\rho)-S(\rho_{A})$ with $S(X)=-Tr X\log_2 X$ the von Neumann entropy of a state $X$, $\rho_A$ is the reduced state associated to the system $A$. $S_{B| \Pi^{A}}(\rho)=\sum\limits_jp_{j}^{A}S(\Pi^{A}_{j}  \rho  \Pi^{A}_{j}  / p_{j}^{A}),$
where $\Pi^{A}_{j}$ is the von Neumann projection operator on the subsystem A and $p_{j}^{A}=Tr(\Pi^{A}_{j}  \rho  \Pi^{A}_{j})$ is the probability with respect to the measurement outcome $j$.

Very recently, Radhakrishnan \emph{et.al} \cite{Radhakrishnan} introduced a generalization of discord for tripartite and multipartite states.
One of the main features of this approach is the use of conditional measurements, where each successive measurement is conditionally related to the previous measurements.
The $(N-1)$-partite measurement is written as
\begin{align}
\Pi_{j_1\cdots j_{N-1}}^{A_1\cdots A_{N-1}}=\Pi_{j_1}^{A_1}\otimes\Pi_{j_2|j_1}^{A_2}\cdots \otimes \Pi_{j_{N-1}|j_1\cdots j_{N-2}}^{A_{N-1}},\nonumber
\end{align}
where $\Pi_{j_1|j_2}^{A_2}$ is a projector on subsystem $A_2$ conditioned on the measurement outcome of $A_1$.
Here the measurements take place in the order $A_1 \rightarrow A_2 \rightarrow \cdots A_{N-1}.$
Such conditioned measurements are essential to take into account all the classical correlations that may exist among the subsystems.
Viewing the measurements as operations to break the quantum correlations, the optimization over all such measurements allows one to recover the pure quantum contributions.
Moreover, there is an obvious asymmetry due the fixed ordering of the measurements.
This asymmetry has similarities with the quantum steering where one also considers measurements on part of a system, while
the aims are somewhat different in that for discord, one minimizes the disturbance due to measurements rather than compares it to a local hidden state theory.
In deed, in some quantum information processing such as one-way quantum computing, there is a definite ordering of measurements,
in compatible with the multipartite discord.

The quantum discord of an N-partite state $\rho$ is defined by
\begin{align}
 D_{A_1;A_2;  \dots ;A_N} (\rho)  = & \min_{ \Pi^{A_1 \dots A_{N-1} }} \Big[
- S_{A_2 \dots A_N|A_1 }(\rho) \nonumber \\
& +  S_{A_2| \Pi^{A_1} } (\rho)  \dots + S_{A_N|  \Pi^{A_1 \dots A_{N-1} } }  (\rho) \Big]
\label{multipartitediscord}
\end{align}
for the measurement ordering $ A_1 \rightarrow A_2 \rightarrow \dots A_{N-1} $. Here
 we have defined $ S_{A_k|\Pi^{A_1 \dots A_{k-1}}} (\rho) =
\sum_{j_1 \dots j_{k-1}} p_{j}^{(k-1)}
S_{A_1 \dots A_k}( \Pi^{(k-1)}_{j}  \rho  \Pi^{(k-1)}_{j}  / p_{j}^{(k-1)} )$
with $ \Pi^{(k)}_{{j}} \equiv \Pi^{A_1 \dots A_{k}}_{j_1 \dots j_{k}} $,
$  p_{j}^{(k)}  = \text{Tr} (  \Pi^{(k)}_{j}  \rho  \Pi^{(k)}_{j}  ) $.

In general, it is difficult to evaluate the quantum discord (\ref{multipartitediscord}) due to the complexity of the optimization. We analyze and evaluate this quantum discord for a family of multi-qubit states, and graphically show the level surfaces of the quantum discord of this family. Moreover, due to the interaction with the environments, bipartite quantum discord may decrease asymptomatically with time \cite{Ban}, and may be also frozen \cite{Mazzola,Maziero,Haikka,Bromley} and decoherence free for certain time. We also study the dynamic behavior of quantum discord for a family of three-qubit and four-qubit states under decoherence. We discover that the multi-qubit quantum discord of some states cannot be destroyed by decoherence in finite time.

The rest of this article is organized as follows. In Sec. \uppercase\expandafter{\romannumeral2}, we calculate analytically the multi-qubit discord for a family of quantum states. We shown that the quantum discord can be classified into three categories. In Sec. \uppercase\expandafter{\romannumeral3}, we investigated the dynamical behavior of the discord for a family of three-qubit and four-qubit states. We discuss and summarize the results in Sec. \uppercase\expandafter{\romannumeral4}.

\section{\bf Quantum discord for multi-qubit systems}

Consider the following family of $N$-qubit states,
\begin{equation}\label{rho}
\rho=\frac{1}{2^N}(I+\sum\limits_{j=1}^3c_j\sigma_j\otimes\cdots\otimes\sigma_j),
\end{equation}
where $\sigma_j$, $j=1,2,3$, are the standard Pauli matrices, $I$ stands for the corresponding identity operator. The motivation to consider the states (\ref{rho}) is that, for $N=2$ (\ref{rho}) reduces to the well-known Bell diagonal states whose famous analytical formulae of quantum discord have been provided by Luo \cite{luo1}, which attracted much attention and resulted in further vital results. For general $N$, these states are highly symmetric and include some generalized GHZ or W states as special ones.

We consider the family of three-qubit state, associated with systems $A$, $B$ and $C$,
\begin{equation}\label{rho3}
\rho=\frac{1}{8}(I+\sum\limits_{j=1}^3c_j\sigma_j\otimes\sigma_j\otimes\sigma_j).
\end{equation}

From (\ref{multipartitediscord}) the quantum discord is given by
\begin{align}
D_{A;B;C} (\rho) = & \min_{\Pi^{AB}  } \Big[  - S_{BC|A} (\rho)
 + S_{B| \Pi^{A}} (\rho) + S_{C|  \Pi^{AB}} (\rho) \Big].
\label{tripartitediscord}
\end{align}
Since $\rho_A=Tr_{BC}(\rho) = \frac{I}{2}$, we have the entropy $S(\rho_A)=1$.
Set $\xi=\sqrt{c_1^2+c_2^2+c_3^2}$. One can verify that
\begin{equation*}
S(\rho)=-4\times\frac{1+\xi}{8}\log_2\frac{1+\xi}{8}
-4\times\frac{1-\xi}{8}\log_2\frac{1-\xi}{8}.
\end{equation*}
Hence,
\begin{align}\label{BC|A}
-S_{BC|A} (\rho)&=-(S(\rho)-S(\rho_A))\nonumber\\
&=\frac{1+\xi}{2}\log_2(1+\xi)+\frac{1-\xi}{2}\log_2(1-\xi)-2.
\end{align}

Denote $\{\Pi_k=|k\rangle\langle k|:k=0,1\}$. The von Neumann measurement on subsystem $A$ is given by $\{A_k=V_A\Pi_kV_A^\dagger:k=0,1\}$, where $V_A=t_AI+i\overrightarrow{y_A}\cdot\overrightarrow{\sigma}$ is the unitary operator with
$t_A\in \mathbb{R}$, $\overrightarrow{y_A}=(y_{A1},y_{A2},y_{A3})\in \mathbb{R}^3$ and
$t_A^2+y_{A1}^2+y_{A2}^2+y_{A3}^2=1$. After the measurement ${A_k}$, the state $\rho$ is going to become the ensemble $\{\rho_k,p_k\}$ with $\rho_k:=\frac{1}{p_k}(A_k\otimes I)\rho(A_k\otimes I)$
and $p_k=Tr(A_k\otimes I)\rho(A_k\otimes I).$
Then we obtain $p_0=p_1=\frac{1}{2}$,
\begin{align}
\rho_0=&\frac{1}{4}V_A\Pi_0V_A^\dagger\otimes(I\otimes I+c_1z_{1}\sigma_1\otimes\sigma_1+c_2z_{2}\sigma_2\otimes\sigma_2\nonumber\\
&+c_3z_{3}\sigma_3\otimes\sigma_3)
\label{rho0}
\end{align}
and
\begin{align}
\rho_1=&\frac{1}{4}V_A\Pi_1V_A^\dagger\otimes(I\otimes I-c_1z_{1}\sigma_1\otimes\sigma_1-c_2z_{2}\sigma_2\otimes\sigma_2\nonumber\\
&-c_3z_{3}\sigma_3\otimes\sigma_3),
\label{rho1}
\end{align}
with
\begin{align}
& z_{1}=2(-t_Ay_{A2}+y_{A1}y_{A3}),\nonumber \\
& z_{2}=2(t_Ay_{A1}+y_{A2}y_{A3}),\nonumber \\
& z_{3}=t_A^2-y_{A1}^2-y_{A2}^2+y_{A3}^2.\nonumber
\end{align}
Thus, we have
$Tr_C(\rho_0)=\frac{1}{2}V_A\Pi_0V_A^\dagger\otimes I$ and  $Tr_C(\rho_1)=\frac{1}{2}V_A\Pi_1V_A^\dagger\otimes I$.
The average entropy of the subsystem $B$ after measure $\Pi^A$ is given by
\begin{equation}\label{B|pi^A}
S_{B|\Pi^A}(\rho)=\frac{1}{2}\times 1+\frac{1}{2}\times 1=1.
\end{equation}

To evaluate $S_{C|  \Pi^{AB}} (\rho)$, one needs to measure the subsystem $B$ under the conditions of the outcomes on measuring $A$. Let
\begin{equation*}
\{B_k^j=V_{B^j}\Pi_kV_{B^j}^\dagger:k=0,1\},~j=0,1,
\label{measureb}
\end{equation*}
be the von Neumann measurement on the subsystem $B$ when the outcome of the measurement on $A$ is $j$ ($j=0,1$), where $V_{B^j}=t_{B^j}I+i\overrightarrow{y_{B^j}}\cdot\overrightarrow{\sigma}$ is the unitary operator with $t_{B^j}\in \mathbb{R}$, $\overrightarrow{y_{B^j}}=(y_{B^j1},y_{B^j2},y_{B^j3})\in \mathbb{R}^3$ and
$t_{B^j}^2+y_{B^j1}^2+y_{B^j2}^2+y_{B^j3}^2=1$.

 If the measurement outcome on system $A$ is $0$,the state after the measurement will be reduced to $\rho_{0}$ given in  (\ref{rho0}). Notice that the subsystems $B$ and $C$ in (\ref{rho0}) is still in a Bell-diagonal state. After performing the measurement
 $\{B_k^0:k=0,1\}$, the state reduces to
 \begin{eqnarray}
\rho_{00}&=&\frac{1}{2}V_A\Pi_0V_A^\dagger\otimes V_{B^0}\Pi_0V_{B^0}^\dagger\otimes(I+c_1z_{1}l_1\sigma_1\nonumber\\
& &+c_2z_{2}l_2\sigma_2+c_3z_{3}l_3\sigma_3),\nonumber
\label{rho00}
\end{eqnarray}
\begin{eqnarray}
\rho_{01}&=&\frac{1}{2}V_A\Pi_0V_A^\dagger\otimes V_{B^0}\Pi_1V_{B^0}^\dagger\otimes(I-c_1z_{1}l_1\sigma_1\nonumber\\
& &-c_2z_{2}l_2\sigma_2-c_3z_{3}l_3\sigma_3),\nonumber
\label{rho01}
\end{eqnarray}
with the probability $p_{00}=p_{01}=\frac{1}{4}$,where
\begin{align}
& l_{1}=2(-t_{B^0}y_{B^02}+y_{B^01}y_{B^03}),\nonumber \\
& l_{2}=2(t_{B^0}y_{B^01}+y_{B^02}y_{B^03}),\nonumber\\
& l_{3}=t_{B^0}^2-y_{B^01}^2-y_{B^02}^2+y_{B^03}^2.\nonumber
\end{align}

If the measurement outcome on system $A$ is $1$, performing the measurement
 $\{B_k^1:k=0,1\}$ on the subsystem $B$ of the state $\rho_1$, we obtain
 \begin{eqnarray}
\rho_{10}&=&\frac{1}{2}V_A\Pi_1V_A^\dagger\otimes V_{B^1}\Pi_0V_{B^1}^\dagger\otimes(I-c_1z_{1}m_1\sigma_1\nonumber\\
& &-c_2z_{2}m_2\sigma_2-c_3z_{3}m_3\sigma_3),\nonumber
\label{rho10}
\end{eqnarray}
\begin{eqnarray}
\rho_{11}&=&\frac{1}{2}V_A\Pi_1V_A^\dagger\otimes V_{B^1}\Pi_1V_{B^1}^\dagger\otimes(I+c_1z_{1}m_1\sigma_1\nonumber\\
& &+c_2z_{2}m_2\sigma_2+c_3z_{3}m_3\sigma_3),\nonumber
\label{rho11}
\end{eqnarray}
with the probability $p_{10}=p_{11}=\frac{1}{4}$,
where
\begin{align}
& m_{1}=2(-t_{B^1}y_{B^12}+y_{B^11}y_{B^13}),\nonumber \\
& m_{2}=2(t_{B^1}y_{B^11}+y_{B^12}y_{B^13}),\nonumber \\
& m_{3}=t_{B^1}^2-y_{B^11}^2-y_{B^12}^2+y_{B^13}^2.\nonumber
\end{align}

The state $\rho_{\Pi^{AB}}$ is given by $\rho_{\Pi^{AB}}=p_{00}\rho_{00}+p_{01}\rho_{01}+p_{10}\rho_{10}+p_{11}\rho_{11}$.
Set $\alpha=\sqrt{c_1^2z_1^2l_1^2+c_2^2z_2^2l_2^2+c_3^2z_3^2l_3^2}$
and $\beta=\sqrt{c_1^2z_1^2m_1^2+c_2^2z_2^2m_2^2+c_3^2z_3^2m_3^2}$.
Then
\begin{align}
S_{C|\Pi^{AB}}(\rho)=&-\frac{1+\alpha}{4}\log_2(1+\alpha)-\frac{1-\alpha}{4}\log_2(1-\alpha)\nonumber\\
&-\frac{1+\beta}{4}\log_2(1+\beta)-\frac{1-\beta}{4}\log_2(1-\beta)\nonumber\\
&+1.\nonumber
\end{align}
It can be directly verified that $z_1^2+z_2^2+z_3^2=1$, $l_1^2+l_2^2+l_3^2=1$, $m_1^2+m_2^2+m_3^2=1$. Denote
\begin{equation}\label{c}
c:=max\{|c_1|,|c_2|,|c_3|\}.
\end{equation}
Then
\begin{equation}\label{alpha}
\alpha\leq\sqrt{|c^2|(|z_1|^2|l_1|^2+|z_2|^2|l_2|^2+|z_3|^2|l_3|^2)}=c,
\end{equation}
and
\begin{equation}\label{beta}
\beta\leq\sqrt{|c^2|(|z_1|^2|m_1|^2+|z_2|^2|m_2|^2+|z_3|^2|m_3|^2)}=c.
\end{equation}

The equality holds in (\ref{alpha}) for the following cases:
(1) if $c=|c_1|$, then $|z_1|=|l_1|=1,z_2=z_3=l_2=l_3=0$. For instance, $|t_A|=|y_{A2}|=|t_{B^0}|=|y_{B^02}|=\frac{1}{\surd2}$ and $y_{A1}=y_{A3}=y_{B^01}=y_{B^03}=0;$
(2) if $c=|c_2|$, then $|z_2|=|l_2|=1,z_1=z_3=l_1=l_3=0$. For example, $|t_A|=|y_{A1}|=|t_{B^0}|=|y_{B^01}|=\frac{1}{\surd2}$ and $y_{A2}=y_{A3}=y_{B^02}=y_{B^03}=0;$
(3) if $c=|c_3|,$ then $|z_3|=|l_3|=1,z_1=z_2=l_1=l_2=0,$ e.g., $y_{A1}=y_{A2}=y_{B^01}=y_{B^02}=0.$
Similarly, one can prove that the equality holds in (\ref{beta}) too for the above cases.
Therefore, we obtain
\begin{align}\label{C|pi^AB}
&min(S_{C|\Pi^{AB}}(\rho))\nonumber\\
=&-\frac{1+c}{2}\log_2(1+c)-\frac{1-c}{2}\log_2(1-c)+1.
\end{align}

From (\ref{BC|A}), (\ref{B|pi^A}) and (\ref{C|pi^AB}), we get the quantum discord
\begin{align}
D_{A;B;C} (\rho) &=  \min_{\Pi^{AB}  } \Big[  - S_{BC|A} (\rho)+ S_{B| \Pi^{A}} (\rho) + S_{C|  \Pi^{AB}} (\rho) \Big]\nonumber\\
&=\frac{1+\xi}{2}\log_2(1+\xi)+\frac{1-\xi}{2}\log_2(1-\xi)\nonumber\\
&\quad-\frac{1+c}{2}\log_2(1+c)-\frac{1-c}{2}\log_2(1-c).\label{cl3}
\end{align}

We now consider the family of four-qubit case,
\begin{equation}\label{rho4}
\rho=\frac{1}{16}(I+\sum\limits_{j=1}^3c_j\sigma_j\otimes\sigma_j\otimes\sigma_j\otimes\sigma_j)
\end{equation}
in systems $A_1$, $A_2$, $A_3$ and $A_4$. The four-qubit quantum discord is given by
\begin{align}
D_{A_1;A_2;A_3;A_4} (\rho) =&  \min_{\Pi^{A_1A_2A_3}  } \Big[  - S_{A_2A_3A_4|A_1} (\rho)\nonumber\\
 &+ S_{A_2| \Pi^{A_1}} (\rho) + S_{A_3|  \Pi^{A_1A_2}} (\rho) \nonumber\\
 &+ S_{A_4|  \Pi^{A_1A_2A_3}} (\rho) \Big].
\label{fourpartitediscord}
\end{align}

For (\ref{rho4}) we have $\rho_{A_1}=Tr_{A_2A_3A_4} (\rho) = \frac{I}{2}$ and the entropy of the subsystem $A_1$ is $S(\rho_{A_1})=1.$ It can be directly verified that
\begin{align}
S(\rho)=&-\frac{1}{4}[(1+c_1-c_2-c_3)\log_2(1+c_1-c_2-c_3)\nonumber\\
&+(1-c_1+c_2-c_3)\log_2(1-c_1+c_2-c_3)\nonumber\\
&+(1-c_1-c_2+c_3)\log_2(1-c_1-c_2+c_3)\nonumber\\
&+(1+c_1+c_2+c_3)\log_2(1+c_1+c_2+c_3)]+4.\nonumber
\end{align}
Therefore,
\begin{align}
&-S_{A_2A_3A_4|A_1} (\rho)\nonumber\\
=&\frac{1}{4}[(1+c_1-c_2-c_3)\log_2(1+c_1-c_2-c_3)\nonumber\\
&+(1-c_1+c_2-c_3)\log_2(1-c_1+c_2-c_3)\nonumber\\
&+(1-c_1-c_2+c_3)\log_2(1-c_1-c_2+c_3)\nonumber\\
&+(1+c_1+c_2+c_3)\log_2(1+c_1+c_2+c_3)]-3.
\end{align}

The von Neumann measurement on the subsystem $A_1$ is given by
$\{A_{1k}=V_{A_1}\Pi_kV_{A_1}^\dagger:~k=0,1\}$,
where $V_{A_1}=t_{A_1}I+i\overrightarrow{y_{A_1}}\cdot\overrightarrow{\sigma}$, with $t_{A_1}\in \mathbb{R}$, $\overrightarrow{y_{A_1}}=(y_{{A_1}1},y_{{A_1}2},y_{{A_1}3})\in \mathbb{R}^3$, and $t_{A_1}^2+y_{{A_1}1}^2+y_{{A_1}2}^2+y_{{A_1}3}^2=1.$

The state $\rho_{\Pi^{A_1}}$ is given by $\rho_{\Pi^{A_1}}=p_0\rho_0+p_1\rho_1$,
where $p_0=p_1=\frac{1}{2}$, and
\begin{align}
\rho_0=&\frac{1}{8}V_{A_1}\Pi_0V_{A_1}^\dagger\otimes(I\otimes I\otimes I+c_1d_{1}\sigma_1\otimes\sigma_1\otimes\sigma_1\nonumber\\
&+c_2d_{2}\sigma_2\otimes\sigma_2\otimes\sigma_2
+c_3d_{3}\sigma_3\otimes\sigma_3\otimes\sigma_3),\nonumber
\end{align}
\begin{align}
\rho_1=&\frac{1}{8}V_{A_1}\Pi_1V_{A_1}^\dagger\otimes(I\otimes I\otimes I-c_1d_{1}\sigma_1\otimes\sigma_1\otimes\sigma_1\nonumber\\
&-c_2d_{2}\sigma_2\otimes\sigma_2\otimes\sigma_2-c_3d_{3}\sigma_3\otimes\sigma_3\otimes\sigma_3),\nonumber
\end{align}
where
\begin{align}
& d_{1}=2(-t_{A_1}y_{{A_1}2}+y_{{A_1}1}y_{{A_1}3}),\nonumber \\
& d_{2}=2(t_{A_1}y_{{A_1}1}+y_{{A_1}2}y_{{A_1}3}),\nonumber \\
& d_{3}=t_{A_1}^2-y_{{A_1}1}^2-y_{{A_1}2}^2+y_{{A_1}3}^2.\nonumber
\end{align}
Thus, we have
$Tr_{A_3A_4}(\rho_0)=\frac{1}{2}V_{A_1}\Pi_0V_{A_1}^\dagger\otimes I$ and $Tr_{A_3A_4}(\rho_1)=\frac{1}{2}V_{A_1}\Pi_1V_{A_1}^\dagger\otimes I$.
The average entropy of the subsystem $A_2$ after the measurement $\Pi^{A_1}$ is given by
$S_{{A_2}|\Pi^{A_1}}(\rho)=1$.

To evaluate $S_{A_3|\Pi^{A_1A_2}} (\rho)$ and $S_{A_4|\Pi^{A_1A_2A_3}} (\rho)$, we need to measure the subsystem $A_2$ based on the measurement outcomes on $A_1$. We obtain
\begin{align}
\rho_{00}=&\frac{1}{4}V_{A_1}\Pi_0V_{A_1}^\dagger\otimes V_{A_2^0}\Pi_0V_{A_2^0}^\dagger\otimes(I\otimes I+c_1d_{1}e_1\sigma_1\otimes\sigma_1\nonumber\\
&+c_2d_{2}e_2\sigma_2\otimes\sigma_2+c_3d_{3}e_3\sigma_3\otimes\sigma_3),\nonumber
\end{align}
\begin{align}
\rho_{01}=&\frac{1}{4}V_{A_1}\Pi_0V_{A_1}^\dagger\otimes V_{A_2^0}\Pi_1V_{A_2^0}^\dagger\otimes(I\otimes I-c_1d_{1}e_1\sigma_1\otimes\sigma_1\nonumber\nonumber\\
&-c_2d_{2}e_2\sigma_2\otimes\sigma_2-c_3d_{3}e_3\sigma_3\otimes\sigma_3),\nonumber
\end{align}
 \begin{align}
\rho_{10}=&\frac{1}{4}V_{A_1}\Pi_1V_{A_1}^\dagger\otimes V_{A_2^1}\Pi_0V_{A_2^1}^\dagger\otimes(I\otimes I-c_1d_{1}f_1\sigma_1\otimes\sigma_1\nonumber\\
&-c_2d_{2}f_2\sigma_2\otimes\sigma_2-c_3d_{3}f_3\sigma_3\otimes\sigma_3),\nonumber
\end{align}
\begin{align}
\rho_{11}=&\frac{1}{4}V_{A_1}\Pi_1V_{A_1}^\dagger\otimes V_{A_2^1}\Pi_1V_{A_2^1}^\dagger\otimes(I\otimes I+c_1d_{1}f_1\sigma_1\otimes\sigma_1\nonumber\\
&+c_2d_{2}f_2\sigma_2\otimes\sigma_2+c_3d_{3}f_3\sigma_3\otimes\sigma_3),\nonumber
\end{align}
where the $k$ in the unitary $\{V_{A_2^k}:k=0,1\}$ is the outcome of the measurement of $A_1$ , and $V_{A_2^k}$ can be written as $V_{A_2^k}=t_{A_2^k}I+i\overrightarrow{y_{A_2^k}}\cdot\overrightarrow{\sigma}$, with $t_{A_2^k}\in \mathbb{R}$, $\overrightarrow{y_{A_2^k}}=(y_{{A_2^k}1},y_{{A_2^k}2},y_{{A_2^k}3})\in \mathbb{R}^3$, and
$t_{A_2^k}^2+y_{{A_2^k}1}^2+y_{{A_2^k}2}^2+y_{{A_2^k}3}^2=1,$
\begin{align}
& e_{1}=2(-t_{A_2^0}y_{A_2^02}+y_{A_2^01}y_{A_2^03}),\nonumber \\
& e_{2}=2(t_{A_2^0}y_{A_2^01}+y_{A_2^02}y_{A_2^03}),\nonumber\\
& e_{3}=t_{A_2^0}^2-y_{A_2^01}^2-y_{A_2^02}^2+y_{A_2^03}^2,\nonumber
\end{align}
\begin{align}
& f_{1}:=2(-t_{A_2^1}y_{A_2^12}+y_{A_2^11}y_{A_2^13}),\nonumber \\
& f_{2}:=2(t_{A_2^1}y_{A_2^11}+y_{A_2^12}y_{A_2^13}),\nonumber \\
& f_{3}:=t_{A_2^1}^2-y_{A_2^11}^2-y_{A_2^12}^2+y_{A_2^13}^2.\nonumber
\end{align}

The state $\rho_{\Pi^{A_1A_2}}$ is given by $\rho_{\Pi^{A_1A_2}}=p_{00}\rho_{00}+p_{01}\rho_{01}+p_{10}\rho_{10}+p_{11}\rho_{11}$.
Thus, we have
$Tr_{A_4}(\rho_{00})=\frac{1}{2}V_{A_1}\Pi_0V_{A_1}^\dagger\otimes V_{A_2^0}\Pi_0V_{A_2^0}^\dagger\otimes I,$
$Tr_{A_4}(\rho_{01})=\frac{1}{2}V_{A_1}\Pi_0V_{A_1}^\dagger\otimes V_{A_2^0}\Pi_1V_{A_2^0}^\dagger\otimes I,$
$Tr_{A_4}(\rho_{10})=\frac{1}{2}V_{A_1}\Pi_1V_{A_1}^\dagger\otimes V_{A_2^1}\Pi_0V_{A_2^1}^\dagger\otimes I,$
$Tr_{A_4}(\rho_{11})=\frac{1}{2}V_{A_1}\Pi_1V_{A_1}^\dagger\otimes V_{A_2^1}\Pi_1V_{A_2^1}^\dagger\otimes I$.
The average entropy of the sub system $A_3$ after the measurement $\Pi^{A_1A_2}$ is given by
$S_{{A_3}|\Pi^{A_1A_2}}(\rho)=4\times(\frac{1}{4}\times 1)=1$.

To evaluate $S_{A_4|  \Pi^{A_1A_2A_3}} (\rho)$, one needs to continue to measure the subsystem $A_3$ based on the measurement outcomes on $A_1$ and $A_2$. We obtain
 \begin{align}
\rho_{000}=&\frac{1}{2}V_{A_1}\Pi_0V_{A_1}^\dagger\otimes V_{A_2^0}\Pi_0V_{A_2^0}^\dagger\otimes V_{A_3^{00}}\Pi_0V_{A_3^{00}}^\dagger\otimes(I\nonumber\\
&+c_1d_1e_1g_1\sigma_1+c_2d_2e_2g_2\sigma_2+c_3d_3e_3g_3\sigma_3),\nonumber
\end{align}
\begin{align}
\rho_{001}=&\frac{1}{2}V_{A_1}\Pi_0V_{A_1}^\dagger\otimes V_{A_2^0}\Pi_0V_{A_2^0}^\dagger\otimes V_{A_3^{00}}\Pi_1V_{A_3^{00}}^\dagger\otimes(I\nonumber\\
&-c_1d_1e_1g_1\sigma_1-c_2d_2e_2g_2\sigma_2-c_3d_3e_3g_3\sigma_3),\nonumber
\end{align}
\begin{align}
\rho_{010}=&\frac{1}{2}V_{A_1}\Pi_0V_{A_1}^\dagger\otimes V_{A_2^0}\Pi_1V_{A_2^0}^\dagger\otimes V_{A_3^{01}}\Pi_0V_{A_3^{01}}^\dagger\otimes(I\nonumber\\
&-c_1d_1e_1h_1\sigma_1-c_2d_2e_2h_2\sigma_2-c_3d_3e_3h_3\sigma_3),\nonumber
\end{align}
\begin{align}
\rho_{011}=&\frac{1}{2}V_{A_1}\Pi_0V_{A_1}^\dagger\otimes V_{A_2^0}\Pi_1V_{A_2^0}^\dagger\otimes V_{A_3^{01}}\Pi_1V_{A_3^{01}}^\dagger\otimes(I\nonumber\\
&+c_1d_1e_1h_1\sigma_1+c_2d_2e_2h_2\sigma_2+c_3d_3e_3h_3\sigma_3),\nonumber
\end{align}
\begin{align}
\rho_{100}=&\frac{1}{2}V_{A_1}\Pi_1V_{A_1}^\dagger\otimes V_{A_2^1}\Pi_0V_{A_2^1}^\dagger\otimes V_{A_3^{10}}\Pi_0V_{A_3^{10}}^\dagger\otimes(I\nonumber\\
&-c_1d_1f_1n_1\sigma_1-c_2d_2f_2n_2\sigma_2-c_3d_3f_3n_3\sigma_3),\nonumber
\end{align}
\begin{align}
\rho_{101}=&\frac{1}{2}V_{A_1}\Pi_1V_{A_1}^\dagger\otimes V_{A_2^1}\Pi_0V_{A_2^1}^\dagger\otimes V_{A_3^{10}}\Pi_1V_{A_3^{10}}^\dagger\otimes(I\nonumber\\
&+c_1d_1f_1n_1\sigma_1+c_2d_2f_2n_2\sigma_2+c_3d_3f_3n_3\sigma_3),\nonumber
\end{align}
\begin{align}
\rho_{110}=&\frac{1}{2}V_{A_1}\Pi_1V_{A_1}^\dagger\otimes V_{A_2^1}\Pi_1V_{A_2^1}^\dagger\otimes V_{A_3^{11}}\Pi_0V_{A_3^{11}}^\dagger\otimes(I\nonumber\\
&+c_1d_1f_1r_1\sigma_1+c_2d_2f_2r_2\sigma_2+c_3d_3f_3r_3\sigma_3),\nonumber
\end{align}
\begin{align}
\rho_{111}=&\frac{1}{2}V_{A_1}\Pi_1V_{A_1}^\dagger\otimes V_{A_2^1}\Pi_1V_{A_2^1}^\dagger\otimes V_{A_3^{11}}\Pi_1V_{A_3^{11}}^\dagger\otimes(I\nonumber\\
&-c_1d_1f_1r_1\sigma_1-c_2d_2f_2r_2\sigma_2-c_3d_3f_3r_3\sigma_3),\nonumber
\end{align}
where the $k$ in the unitary $\{V_{A_3^{ku}}:k=0,1;u=0,1\}$ is the outcome of the measurement of $A_1$ , and $u$ is the outcome of the measurement of $A_2$.
Denote
\begin{align}
\mu_1=\sqrt{c_1^2d_1^2e_1^2g_1^2+c_2^2d_2^2e_2^2g_2^2+c_3^2d_3^2e_3^2g_3^2},\nonumber\\
\mu_2=\sqrt{c_1^2d_1^2e_1^2h_1^2+c_2^2d_2^2e_2^2h_2^2+c_3^2d_3^2e_3^2h_3^2},\nonumber\\
\mu_3=\sqrt{c_1^2d_1^2f_1^2n_1^2+c_2^2d_2^2f_2^2n_2^2+c_3^2d_3^2f_3^2n_3^2},\nonumber\\
\mu_4=\sqrt{c_1^2d_1^2f_1^2r_1^2+c_2^2d_2^2f_2^2r_2^2+c_3^2d_3^2f_3^2r_3^2}.\nonumber
\end{align}
We have
\begin{align}
&S_{A_4|\Pi^{A_1A_2A_3}}(\rho)\nonumber\\
=&-\frac{1+\mu_1}{8}\log_2(1+\mu_1)-\frac{1-\mu_1}{8}\log_2(1-\mu_1)\nonumber\\
&-\frac{1+\mu_2}{8}\log_2(1+\mu_2)-\frac{1-\mu_2}{8}\log_2(1-\mu_2)\nonumber\\
&-\frac{1+\mu_3}{8}\log_2(1+\mu_3)-\frac{1-\mu_3}{8}\log_2(1-\mu_3)\nonumber\\
&-\frac{1+\mu_4}{8}\log_2(1+\mu_4)-\frac{1-\mu_4}{8}\log_2(1-\mu_4)+1.\nonumber
\end{align}

It can be directly verified that $d_1^2+d_2^2+d_3^2=1,$ $e_1^2+e_2^2+e_3^2=1,$ $f_1^2+f_2^2+f_3^2=1,$ $g_1^2+g_2^2+g_3^2=1,$  $h_1^2+h_2^2+h_3^2=1,$ $n_1^2+n_2^2+n_3^2=1$ and $r_1^2+r_2^2+r_3^2=1.$ Since $\mu_1\leq c$, $\mu_2\leq c$, $\mu_3\leq c$ and $\mu_4\leq c$,
we obtain
\begin{align}
&min(S_{A_4|\Pi^{A_1A_2A_3}}(\rho))\nonumber\\
=&-\frac{1+c}{2}\log_2(1+c)-\frac{1-c}{2}\log_2(1-c)+1.
\end{align}
By the definition of the four-qubit quantum discord (\ref{fourpartitediscord}), we get
\begin{align}\label{cl4}
&D_{A_1;A_2;A_3;A_4} (\rho) \nonumber\\
=&\frac{1}{4}[(1+c_1-c_2-c_3)\log_2(1+c_1-c_2-c_3)\nonumber\\
&+(1-c_1+c_2-c_3)\log_2(1-c_1+c_2-c_3)\nonumber\\
&+(1-c_1-c_2+c_3)\log_2(1-c_1-c_2+c_3)\nonumber\\
&+(1+c_1+c_2+c_3)\log_2(1+c_1+c_2+c_3)]\nonumber\\
&-\frac{1+c}{2}\log_2(1+c)-\frac{1-c}{2}\log_2(1-c).
\end{align}

From the results of three-qubit and four-quibt states, we can prove the following conclusion for general $N$-qubit case.

\begin{theorem}
For the family of $N$-qubit states (\ref{rho}), we have the quantum discord:

(1) if $N=2v+1$, $v\in \mathbf{N^+}$,
\begin{align}\label{class1}
&D_{A_1;A_2;  \dots ;A_{2v+1}} (\rho) \nonumber\\
=&\frac{1+\xi}{2}\log_2(1+\xi)+\frac{1-\xi}{2}\log_2(1-\xi)\nonumber\\
&-\frac{1+c}{2}\log_2(1+c)-\frac{1-c}{2}\log_2(1-c),
\end{align}
where $\xi=\sqrt{c_1^2+c_2^2+c_3^2},$ $c=max\{|c_1|,|c_2|,|c_3|\};$

(2) if $N=4v-2$, $v\in \mathbf{N^+}$,
\begin{align}\label{class2}
&D_{A_1;A_2;  \dots ;A_{4v-2}} (\rho) \nonumber\\
=&\frac{1}{4}[(1-c_1-c_2-c_3)\log_2(1-c_1-c_2-c_3)\nonumber\\
&+(1-c_1+c_2+c_3)\log_2(1-c_1+c_2+c_3)\nonumber\\
&+(1+c_1-c_2+c_3)\log_2(1+c_1-c_2+c_3)\nonumber\\
&+(1+c_1+c_2-c_3)\log_2(1+c_1+c_2-c_3)]\nonumber\\
&-\frac{1+c}{2}\log_2(1+c)-\frac{1-c}{2}\log_2(1-c);
\end{align}
(3) if $N=4v$, $v\in \mathbf{N^+}$,
\begin{align}\label{class3}
&D_{A_1;A_2;  \dots ;A_{4v}} (\rho) \nonumber\\
=&\frac{1}{4}[(1+c_1-c_2-c_3)\log_2(1+c_1-c_2-c_3)\nonumber\\
&+(1-c_1+c_2-c_3)\log_2(1-c_1+c_2-c_3)\nonumber\\
&+(1-c_1-c_2+c_3)\log_2(1-c_1-c_2+c_3)\nonumber\\
&+(1+c_1+c_2+c_3)\log_2(1+c_1+c_2+c_3)]\nonumber\\
&-\frac{1+c}{2}\log_2(1+c)-\frac{1-c}{2}\log_2(1-c).
\end{align}
\end{theorem}

\begin{proof}
If $N=2v+1$, we have $Tr_{A_2A_3\dots A_{2v+1}} (\rho) = \frac{I}{2}$, and $S_{A_1}(\rho)=1.$
Let $\lambda$ be the eigenvalues of $\rho$. From the characteristic equation $det|\rho_{2v+1}-\lambda I|=0$, we get
\begin{align*}
[\frac{(1-2^{2v+1}\lambda)^2-c_1^2-c_2^2-c_3^2}{2^{(4v+2)}}]^{2^{2v}}=0.
\end{align*}
The $2^{2v}$ eigenvalues are given by $\frac{1}{2^{2v+1}}(1-\sqrt{c_1^2+c_2^2+c_3^2})$  and  $\frac{1}{2^{2v+1}}(1+\sqrt{c_1^2+c_2^2+c_3^2})$, respectively.
One can verify that
\begin{align}
&-S_{A_2;A_3;  \dots ;A_{2v+1}|A_1} (\rho)\nonumber\\
=&\frac{1+\xi}{2}\log_2(1+\xi)+\frac{1-\xi}{2}\log_2(1-\xi)-2v.
\label{prove11}
\end{align}
We obtain $S_{A_{k}|\Pi^{A_1A_2 \dots A_{k-1}}}(\rho)=1$, where  $k=2,\cdots,2v$, namely,
\begin{align}
&S_{A_2|\Pi^{A_1}}(\rho)=S_{A_3|\Pi^{A_1A_2}}(\rho)=\dots \nonumber\\
=&S_{A_{2v}|\Pi^{A_1A_2 \dots A_{2v-1}}}(\rho)=1.
\label{prove12}
\end{align}
Hence,
\begin{align}
&min(S_{A_{2v+1}|\Pi^{A_1A_2 \dots A_{2v}}}(\rho))\nonumber\\
=&-\frac{1+c}{2}\log_2(1+c)-\frac{1-c}{2}\log_2(1-c)+1.
\label{prove13}
\end{align}
By the definition (\ref{multipartitediscord}), we obtain (\ref{class1}).

If $N=4v-2$, we have $Tr_{A_2A_3\dots A_{4v-2}} (\rho) = \frac{I}{2}$ and $S_{A_1}(\rho)=1.$
$\rho$ has $2^{4v-4}$ eigenvalues given by $\frac{1}{2^{4v-2}}(1-c_1-c_2-c_3)$,   $\frac{1}{2^{4v-2}}(1-c_1+c_2+c_3)$, $\frac{1}{2^{4v-2}}(1+c_1-c_2+c_3)$ and
$\frac{1}{2^{4v-2}}(1+c_1+c_2-c_3)$, respectively.
Then we obtain
$-S_{A_2;A_3;  \dots ;A_4v-2} (\rho)=\frac{1}{4}[(1-c_1-c_2-c_3)\log_2(1-c_1-c_2-c_3)+(1-c_1+c_2+c_3)\log_2(1-c_1+c_2+c_3)
+(1+c_1-c_2+c_3)\log_2(1+c_1-c_2+c_3)+(1+c_1+c_2-c_3)\log_2(1+c_1+c_2-c_3)]-4v+3.$
The entropy after the measurement is the same as (\ref{prove12}) and (\ref{prove13}).
Therefore, we obtain (\ref{class2}). (\ref{class3}) is similarly proved.
\end{proof}

From the Theorem one has that for five-qubit states,
\begin{align}
&D_{A_1;A_2;A_3;A_4;A_5} (\rho) \nonumber\\
=&\frac{1+\xi}{2}\log_2(1+\xi)+\frac{1-\xi}{2}\log_2(1-\xi)\nonumber\\
&-\frac{1+c}{2}\log_2(1+c)-\frac{1-c}{2}\log_2(1-c),\label{cl5}
\end{align}
and for six-qubit states
\begin{align}
&D_{A_1;A_2;A_3;A_4;A_5;A_6} (\rho)  \nonumber\\
=&\frac{1}{4}[(1-c_1-c_2-c_3)\log_2(1-c_1-c_2-c_3)\nonumber\\
&+(1-c_1+c_2+c_3)\log_2(1-c_1+c_2+c_3)\nonumber\\
&+(1+c_1-c_2+c_3)\log_2(1+c_1-c_2+c_3)\nonumber\\
&+(1+c_1+c_2-c_3)\log_2(1+c_1+c_2-c_3)]\nonumber\\
&-\frac{1+c}{2}\log_2(1+c)-\frac{1-c}{2}\log_2(1-c).\label{cl6}
\end{align}
Interestingly, the result (\ref{cl5}) is equivalent to the three-qubit quantum discord (\ref{cl3}), while (\ref{cl6}) is equivalent to the two-qubit quantum discord given by Luo \cite{luo1}.

Fig. 1 shows the level surfaces of discord for $D(\rho)=0.03,0.15$ and $0.55$.
The three figures $(F_{21})$, $(F_{22})$ and $(F_{23})$ in the second row of Fig. 1 are for
$4v-2$-qubit states, which are in consist with the ones given in \cite{lang} for two-qubit states. For small discord, $D(\rho)=0.03$ and $0.15$, the level surfaces are centrally symmetric, consisting of three intersecting ''tubes" along the three coordinate axes.
For larger discord value $0.55$, these intersecting tubes expand until only a few vertices remained, where $(F_{13})$ has level surfaces in eight corners, while $(F_{23})$ and $(F_{33})$ have only four corners left.

\begin{widetext}
\begin{figure*}[tb]
\begin{center}
\includegraphics[width=15.5cm]{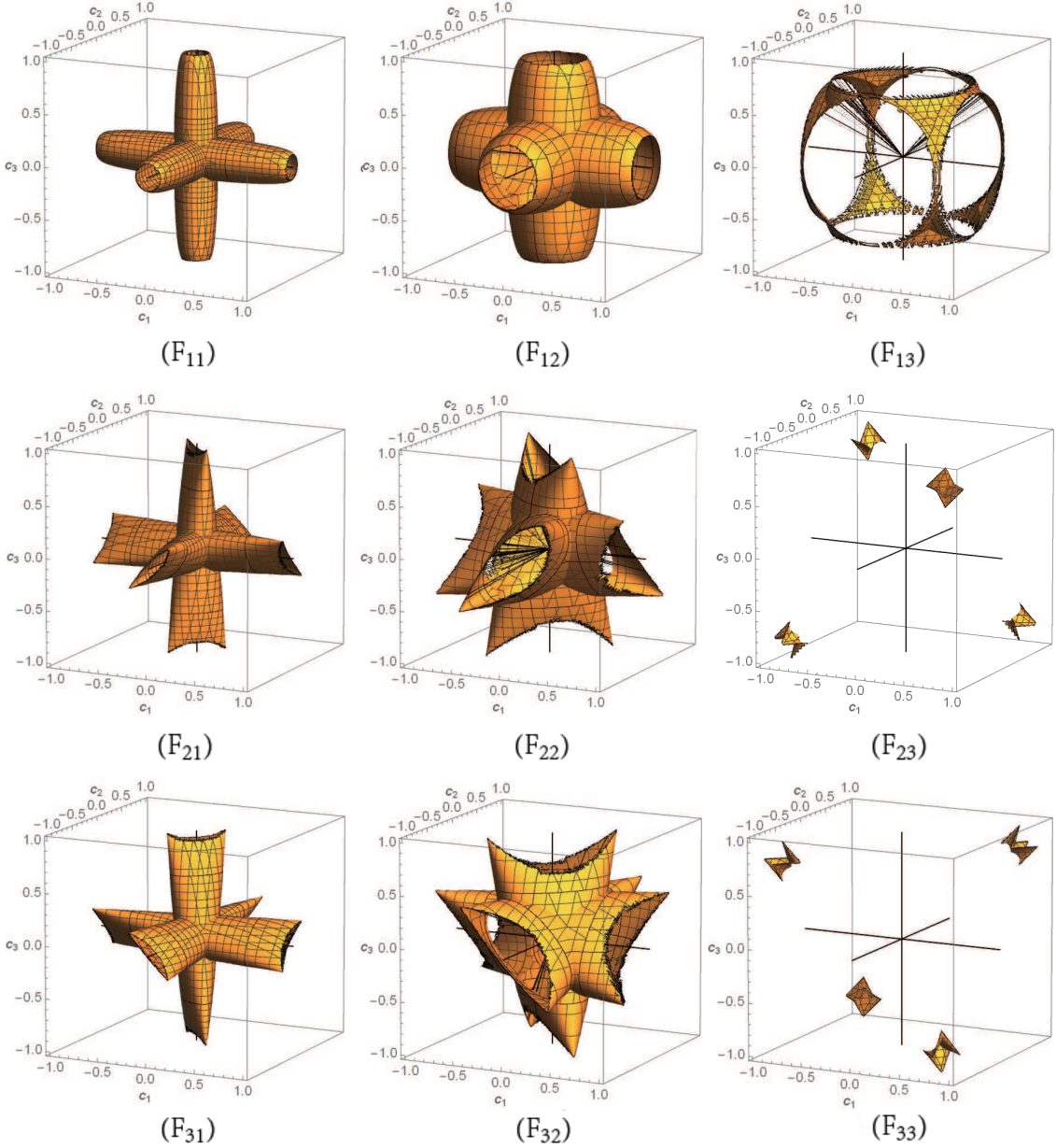}
\caption{Level surfaces of constant discord.
$N=2v+1$ for figures $(F_{11})$, $(F_{12})$ and $(F_{13})$ with $D(\rho)=0.03,~0.15$ and $0.55$, respectively.
$N=4v-2$ for figures $(F_{21})$, $(F_{22})$ and $(F_{23})$ with $D(\rho)=0.03,~0.15$ and $0.55$, respectively.
$N=4v$ for figures $(F_{31})$, $(F_{32})$ and $(F_{33})$ with $D(\rho)=0.03,~0.15$ and $0.55$, respectively.}
\end{center}
\label{Fig:1}
\end{figure*}
\end{widetext}

\section{\bf Dynamics of quantum discord under local nondissipative channels}\label{dynamics}
It has been discovered that for two-qubit states, the quantum discord is invariant under some decoherence channels in a finite time interval\cite{Mazzola,Li}. To verify if such phenomena still exist in multi-qubit systems, we consider that the states $\rho$ (\ref{rho3}) and (\ref{rho4}) under the phase flip channel, with the Kraus operators
$\Gamma_0^{(A_1)}=$ diag$(\sqrt{1-p/2},\sqrt{1-p/2})\otimes I \otimes \dots \otimes I$, $\Gamma_1^{(A_1)}=$ diag$(\sqrt{p/2},-\sqrt{p/2})\otimes I \otimes \dots \otimes I$, $\dots$,
$\Gamma_0^{(A_N)}= I \otimes \dots \otimes I\otimes$ diag$(\sqrt{1-p/2},\sqrt{1-p/2})$, $\Gamma_1^{(A_N)}= I \otimes \dots \otimes I\otimes$ diag$(\sqrt{p/2},-\sqrt{p/2})$, where $N=3,4$, $p=1-\exp(-\gamma t)$, $\gamma$ is the phase damping rate.

Let $\varepsilon(\cdot)$ represent the operator of decoherence. For the three-qubit state (\ref{rho3}) under the phase flip channel, we have
\begin{align}
\varepsilon(\rho)=&\frac{1}{8}(I\otimes I\otimes I+(1-p)^3c_1\sigma_1\otimes\sigma_1\otimes\sigma_1\nonumber\\
&+(1-p)^3c_2\sigma_2\otimes\sigma_2\otimes\sigma_2+c_3\sigma_3\otimes\sigma_3\otimes\sigma_3).      \label{Eq:epsilonrho3}
\end{align}
From (\ref{cl3}), we obtain
\begin{align}
&D_{A_1;A_2;A_3}(\varepsilon(\rho))\nonumber\\
=&\frac{1+\delta}{2}\log_2(1+\delta)+\frac{1-\delta}{2}\log_2(1-\delta)\nonumber\\
&-\frac{1+\theta}{2}\log_2(1+\theta)-\frac{1-\theta}{2}\log_2(1-\theta),
\end{align}
where $\delta=\sqrt{(1-p)^6c_1^2+(1-p)^6c_2^2+c_3^2}, $ $\theta=max\{|(1-p)^3c_1|,|(1-p)^3c_2|,|c_3|\}. $
Notice that the derivative of $D_{A_1;A_2;A_3}(\varepsilon(\rho))$ with respect to $p$ is always less than $0$.

{\sf [Proof]}
From (28), in general, the derivative of $D_{A_1;A_2;A_3}(\varepsilon(\rho))$ can be cast as
\begin{align}
D'_{A_1;A_2;A_3}(\varepsilon(\rho))=\frac{1}{2}[\theta'\log_2(\frac{1-\theta}{1+\theta})-\delta'\log_2(\frac{1-\delta}{1+\delta})].\nonumber
\end{align}
In particular for
\begin{align}
\delta=\sqrt{(1-p)^6c_1^2+(1-p)^6c_2^2+c_3^2}>0,\nonumber
\end{align}
\begin{align}
\theta=max\{|(1-p)^3c_1|,|(1-p)^3c_2|,|c_3|\}>0.\nonumber
\end{align}
This implies that
\begin{align}
\log_2(\frac{1-\theta}{1+\theta})<0
\quad and \quad
\log_2(\frac{1-\delta}{1+\delta})<0.\nonumber
\end{align}
we have that
\begin{align}
\delta'=-\frac{3(c_1^2+c_2^2)}{\delta}(1-p)^5,\nonumber
\end{align}
given that $0<(1-p)^5<1$ then $\delta'<0.$

If $\theta=|c_3|$ then
\begin{align}
D'_{A_1;A_2;A_3}(\varepsilon(\rho))=\frac{-\delta'}{2}\log_2(\frac{1-\delta}{1+\delta})<0.\nonumber
\end{align}
And for $\theta=|(1-p)^3c_1|$,  in this case $\theta=(1-p)^3|c_1|$  because $0<1-p<1$  then
\begin{align}
\theta'=-3(1-p)^2 |c_1|<0.\nonumber
\end{align}
Then
\begin{align}
D'_{A_1;A_2;A_3}(\varepsilon(\rho))
=&\frac{3}{2}[\frac{(c_1^2+c_2^2)}{\delta}(1-p)^5\log_2(\frac{1-\delta}{1+\delta})\nonumber\\
&-(1-p)^2|c_1|\log_2(\frac{1-\theta}{1+\theta})].\nonumber
\end{align}
Here, we assume that $D_{A_1;A_2;A_3}(\varepsilon(\rho))$ is monotonically decreasing, i.e., $D'_{A_1;A_2;A_3}(\varepsilon(\rho))<0.$ In order to show these, it must satisfy that
\begin{align}
\frac{(c_1^2+c_2^2)}{\delta}(1-p)^5\log_2(\frac{1-\delta}{1+\delta})<(1-p)^2|c_1|\log_2(\frac{1-\theta}{1+\theta}),\nonumber
\end{align}
by multiplying by the positive numbers $(1-p)$ and $\delta$
\begin{align}
(c_1^2+c_2^2)(1-p)^6\log_2(\frac{1-\delta}{1+\delta})<\delta\theta\log_2(\frac{1-\theta}{1+\theta}).\nonumber
\end{align}
Therefore, $D'_{A_1;A_2;A_3}(\varepsilon(\rho))<0, $ which means that $D_{A_1;A_2;A_3}(\varepsilon(\rho))$ is monotonically decreasing.
$\Box$

Hence, the frozen phenomenon of quantum discord dos not exist for three-qubit states under the phase flip channel.
Since the quantum discord of three-qubit and $(2v+1)$-qubit are the same, the odd-qubit systems do not exhibit frozen phenomenon of quantum discord under the phase flip channel.
For instance, take $c_1=\frac{4}{5}$, $c_2=\frac{c_1}{2}$ and $c_3=\frac{1}{2}$ in the initial state, the dashed line in Fig. 2 shows the dynamic behavior of the quantum discord under the phase flip channel.

\begin{figure}[h]
\scalebox{2.0}{\includegraphics[width=3.9cm]{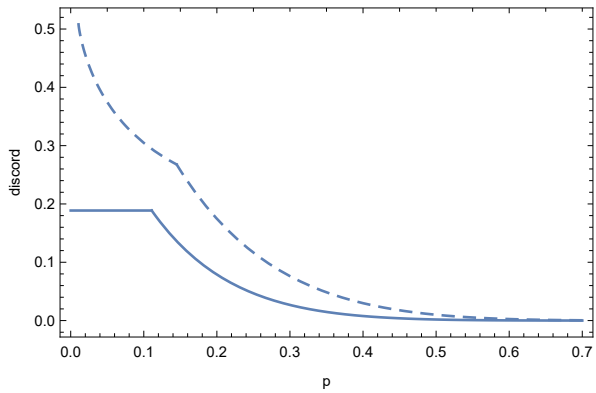}}
\caption{Quantum discord of the three-qubit state (dashed line) and quantum discord of the four-qubit state (solid line) under phase flip channel for $c_1=\frac{4}{5}$, $c_2=\frac{c_1}{2}$, $c_3=\frac{1}{2}$.}
\label{transition}
\end{figure}

For $N=4$, the state $\rho$ under the phase flip channel is give by
\begin{align}
\varepsilon(\rho)=&\frac{1}{16}(I\otimes I\otimes I\otimes I
+(1-p)^4c_1\sigma_1\otimes\sigma_1\otimes\sigma_1\otimes\sigma_1\nonumber\\
&+(1-p)^4c_2\sigma_2\otimes\sigma_2\otimes\sigma_2\otimes\sigma_2\nonumber\\
&+c_3\sigma_3\otimes\sigma_3\otimes\sigma_3\otimes\sigma_3).      \label{Eq:epsilonrho4}
\end{align}
Noting that $c_3$ is independent on time, we consider the case that $c_2=c_1c_3,~ -1\leq c_3\leq 1$. Then we have $|c_2|\leq|c_1|$ for any $p$. From (\ref{cl4}) we obtain the quantum discord
\begin{align}
&D_{A_1;A_2;A_3;A_4} (\varepsilon(\rho))\nonumber\\
=&\frac{1+c_3}{2}\log_2(1+c_3)+\frac{1-c_3}{2}\log_2(1-c_3)\nonumber\\
&+\frac{1+(1-p)^4c_1}{2}\log_2(1+(1-p)^4c_1)\nonumber\\
&+\frac{1-(1-p)^4c_1}{2}\log_2(1-(1-p)^4c_1)\nonumber\\
&-\frac{1+\sigma}{2}\log_2(1+\sigma)-\frac{1-\sigma}{2}\log_2(1-\sigma),\nonumber
\end{align}
where $\sigma=max\{|(1-p)^4c_1|, |c_3|\}. $

When $max\{|(1-p)^4c_1|,|c_3|\}=|(1-p)^4c_1|$, we have
\begin{align}
&D_{A_1;A_2;A_3;A_4} (\varepsilon(\rho))\nonumber\\
=&\frac{1+c_3}{2}\log_2(1+c_3)+\frac{1-c_3}{2}\log_2(1-c_3),\nonumber
\end{align}
$D_{A_1;A_2;A_3;A_4} (\varepsilon(\rho))$ is constant under the decoherence channel during the time interval. Otherwise,
\begin{align}
&D_{A_1;A_2;A_3;A_4} (\varepsilon(\rho))\nonumber\\
=&\frac{1+(1-p)^4c_1}{2}\log_2(1+(1-p)^4c_1)\nonumber\\
&+\frac{1-(1-p)^4c_1}{2}\log_2(1-(1-p)^4c_1),\nonumber
\end{align}
which monotonically decreases to zero.

Therefore, to calculate the quantum discord, we need to determine the magnitude of $|(1-p)^4c_1|$ and $|c_3|$. If for $|c_1|>|c_3|$ there exist $0\leq p_0\leq 1$ such that $max\{|(1-p)^4c_1|,|c_3|\}=|(1-p)^4c_1|$
for $0\leq p\leq p_0,$  and $max\{|(1-p)^4c_1|,|c_3|\}=|c_3|$ for $ p_0\leq p\leq 1$,
then $D_{A_1;A_2;A_3;A_4} (\varepsilon(\rho))$ remains unchanged first, and then monotonicity goes down to zero.
As an example, set $c_1=\frac{4}{5}$, $c_2=\frac{c_1}{2}$ and $c_3=\frac{1}{2}$.
The solid line in Fig. 2 shows the dynamic behavior of quantum discord under the phase flip channel. A sudden transition of quantum discord happens at $p=0.11086.$
The frozen phenomenon of quantum discord exists for the four-qubit states under the phase flip channel, while for the case of three-qubit states, such phenomenon does not exist.

In \cite{Li}, it has been shown that the frozen phenomenon of quantum discord also exists when the phase noise acts on two-qubit states.
Since the quantum discord of two-qubit and $(4v-2)$-qubit are the same, and the four-qubit quantum discord is equal to that of $(4v)$-qubit states,
the even-qubit systems exhibit frozen phenomenon of quantum discord under the phase flip channel, while the odd-qubit systems not.

\section{\bf summary}\label{discuss}
We have studied the quantum discord for a family to multi-qubit states. Analytical formulae have been derived in detail for $(2 v+1)$, $(4v-2)$ and $(4v)$-qubit states. The level surfaces of quantum discord have been depicted. It has been shown that under the phase flip channel the quantum discord could still keep constant in a certain time interval for the even-qubit systems, but not for odd-qubit systems. Our results may highlight further investigations on multipartite quantum discord and their applications in quantum information processing.

\bigskip
\noindent {\bf Acknowledgments}
This work is supported by NSFC under numbers 11765016,
12075159, 12065021, 11905131, Key Project of Beijing Municipal Commission of Education (KZ201810028042), Beijing Natural Science Foundation (Z190005), Academy for Multidisciplinary Studies, Capital Normal University, and Shenzhen Institute for Quantum Science and Engineering, Southern University
of Science and Technology, Shenzhen 518055, China (No. SIQSE202001). X.B.L. was supported
by the Natural Science Foundation of Jiangxi Province (Grant No. 20202BAB201010).

B.L and C.L.Z contributed equally to this work.

\end{document}